\documentclass[11pt]{article}

\usepackage{amsmath}
\usepackage{amssymb}
\usepackage{amsfonts}
\usepackage{amsthm}  %% remove it for LNCS
\usepackage{graphicx}
\usepackage{fullpage}
\usepackage{hyperref}
\usepackage{cite}
\usepackage{enumitem}
\usepackage{algorithm}
\usepackage[noend]{algpseudocode}

\newcommand{\argmin}{\mathop{\rm arg~min}\limits}

\newtheorem{theorem}{Theorem}[section]
\newtheorem{lemma}[theorem]{Lemma}

\newtheorem{corollary}[theorem]{Corollary}
\newtheorem{claim}[theorem]{Claim}
\newtheorem{observation}[theorem]{Observation}

\newtheorem{example}[theorem]{Example}
\newtheorem{result}[theorem]{Result}

\makeatletter
\newcommand{\figcaption}[1]{\def\@captype{figure}\caption{#1}}
\newcommand{\tblcaption}[1]{\def\@captype{table}\caption{#1}}
\makeatother

\title{EFX Allocations for Indivisible Chores: \\ Matching-Based Approach}

\author{Yusuke Kobayashi\thanks{Research Institute for Mathematical Sciences, Kyoto University.
E-mail: yusuke@kurims.kyoto-u.ac.jp}
\and
Ryoga Mahara\thanks{Department of Mathematical Informatics, University of Tokyo.
E-mail: mahara@mist.i.u-tokyo.ac.jp}
\and
Souta Sakamoto\thanks{Research Institute for Mathematical Sciences, Kyoto University.}
}

\date{}

\begin{document}
\maketitle
\begin{abstract}
One of the most important topics in discrete fair division is whether an EFX allocation exists for any instance. 
Although the existence of EFX allocations is a standing open problem for both goods and chores, the understanding of the existence of EFX allocations for chores is less established compared to goods. We study the existence of EFX allocation for chores under the assumption that all agent's cost functions are additive. Specifically, we show the existence of EFX allocations for the following three cases: (i) the number of chores is at most twice the number of agents, (ii) the cost functions of all agents except for one are identical ordering, and (iii) the number of agents is three and each agent has a personalized bi-valued cost function.
Furthermore, we provide a polynomial time algorithm to find an EFX allocation for each case.

\end{abstract}

\section{Introduction}
Fair division theory has significant attention across various fields, including economics, mathematics, and computer science. The classic problem of fairly dividing divisible resources, also known as the cake-cutting problem, dates back to the 1940s \cite{Steinhaus} and has a long history\cite{brams1996fair, robertson1998fair, moulin2004fair, brandt2016handbook}. 
In contrast, the fair allocation of indivisible items has been a topic of active research in recent decades (see surveys \cite{ijcai2022p756, aziz2022algorithmic}).
Given a set $N=\{1,2,\ldots, n\}$ of $n$ agents and a set $M$ of $m$ items, the goal is to allocate $M$ to $N$ in a fair manner.
We refer to items as {\it goods} if they are beneficial, such as cars and smartphones, and as {\it chores} if they are burdens, such as housework and teaching duties.
For the case of goods, each agent $i\in N$ has a valuation function $v_i: 2^M \rightarrow \mathbb{R}_{\ge 0}$, while for the case of chores, each agent $i\in N$ has a cost function $c_i: 2^{M} \rightarrow \mathbb{R}_{\ge 0}$.
In general, $v_i$ and $c_i$ are assumed to be monotone non-decreasing.

One of the most popular and well-studied fairness notions is {\it envy-freeness} (EF)\cite{foley1966resource}.
Informally speaking, an allocation is called EF if each agent prefers their own bundle to be at least as good as that of any other agent. 
In the case of divisible resources, an EF allocation always exists for both goods and chores\cite{aziz2016discrete, dehghani2018envy}, while in the case of indivisible items, it may not exist (for example, dividing one item among two agents).
This has motivated researchers to consider relaxing notions of EF, such as EF1, EFX, and other related notions.

\paragraph{Envy-freeness up to {\it any} item (EFX) for goods}
One of the most well-studied relaxed notions of EF is EFX, which was proposed by Caragiannis et al.~\cite{caragiannis2016unreasonable, caragiannis2019unreasonable}.
An allocation $A=(A_1, A_2, \ldots , A_n)$ of goods is called {\it EFX} if for all pairs of agents $i$ and $j$, and for any $g \in A_j$, it holds that $v_i(A_i) \ge v_i(A_j\setminus \{g\})$.
In other words, each agent $i$ prefers their own bundle to be at least as good as the bundle of agent $j$ after the removal of {\it any} good in $j$'s bundle.
EFX is regarded as the best analog of envy-freeness in the discrete fair division: 
Caragiannis et al.~\cite{caragiannis2019envy} remarked that {\it ``Arguably, EFX is the best fairness analog of envy-freeness for indivisible items.''} 
However, the existence of EFX allocations is not well understood, and it is recognized as a significant open problem in the field of fair division. 
Procaccia~\cite{procaccia2020technical} remarked that {\it ``This fundamental and deceptively accessible question (EFX existence) is open. In my view, it is the successor of envy-free cake cutting as fair division’s biggest problem.''}

There has been a significant amount of research to investigate the existence of EFX allocations in various special cases:
For general valuations, i.e., each valuation function $v_i$ is only assumed to be (i) {\it normalized}: $v_i(\emptyset) = 0$ and (ii) {\it monotone}:  $S\subseteq T$ implies $v_i(S) \le v_i(T)$ for any $S,T \subseteq M$, Plaut and Roughgarden~\cite{plaut2020almost} showed that an EFX allocation always exists when there are two agents, or when all agents have identical valuations.
This result was extended to the case where all agents have one of two general valuations~\cite{mahara2021extension}.
In addition, Mahara showed the existence of EFX allocations when $m\le n+3$ in \cite{mahara2021extension}.
Chaudhury et al.~\cite{chaudhury2020efx} showed that an EFX allocation always exists when $n=3$.
This result was extended to the case where all agents have nice-cancelable valuations, which generalize additive valuations~\cite{berger2022almost}.
A further generalization was obtained by Akrami et al.~\cite{akrami2022efx}, who showed the existence of EFX allocations when there are three agents with two general valuations and one MMS-feasible valuation, which generalizes nice-cancelable valuations.
The existence of EFX allocations of goods remains open even when there are four agents with additive valuations.

\paragraph{Envy-freeness up to {\it any} item (EFX) for chores}
We understand even less about the existence of EFX allocations for chores than for goods.
By analogy to goods, EFX for chores can be defined as follows.
An allocation $A=(A_1, A_2, \ldots , A_n)$ of chores is called {\it EFX} if for all pairs of agents $i$ and $j$, and for any $e \in A_i$, it holds that $c_i(A_i\setminus \{e\}) \le c_i(A_j)$.
In other words, each agent $i$ prefers their own bundle to be at least as good as the bundle of agent $j$ after the removal of {\it any} chore in $i$'s bundle.
Chen and Liu~\cite{chen2020fairness} showed the existence of EFX allocations for $n$ agents with identical valuations and cost functions in the case where goods and chores are mixed.
Li et al.~\cite{li2022almost} showed that an EFX allocation for chores always exists when all agents have an identical ordering cost function by using the top-trading envy graph, which is a tool modified from the envy graph.
Gafni et al.~\cite{gafni2021unified} showed the existence of EFX allocations when each agent has a leveled cost function where a larger set of chores is always more burdensome than a smaller set.
Zhou and Wu~\cite{zhou2021approximately} showed a positive result when $n=3$ and the bi-valued instances, in which each agent has at most two cost values on the chores.
Yin and Mehta~\cite{yin2022envy} showed that if two of the three agents' functions have an identical ordering of chores, are additive, and evaluate every non-singleton set of chores as more burdensome than any single chore, then an EFX allocation exists.
The existence of EFX allocations of chores remains open even when there are three agents with additive valuations.
\subsection{Our Results}
We study the existence of EFX allocations of chores for some special cases under the assumption that each agent has an additive cost function.
We show that an EFX allocation always exists in each of the following three cases:
\begin{result}[Theorem~\ref{thm:2n}]\label{res:1}
There exists an EFX allocation of chores when $m\le 2n$ and each agent has an additive cost function.
Moreover, we can find an EFX allocation in polynomial time.
\end{result}

\begin{result}[Theorem~\ref{thm:ordering}]\label{res:2}
There exists an EFX allocation of chores when $n-1$ agents have identical ordering cost functions.
Moreover, we can find an EFX allocation in polynomial time.
\end{result}

\begin{result}[Theorem~\ref{thm:bivalued}]\label{res:3}
There exists an EFX allocation of chores when $n=3$ and each agent has a personalized bi-valued cost function.
Moreover, we can find an EFX allocation in polynomial time.
\end{result}

The first result is the case where the number of chores is small compared to the number of agents.
If $m$ is at most $n$, then there is an obvious EFX allocation (each agent should be allocated at most one chore).
To the best of our knowledge, Result~\ref{res:1} is the first nontrivial result for a small number of chores.
Interestingly, as mentioned before, for the case of goods, positive results are shown only when $m$ is at most $n+3$~\cite{mahara2021extension}.

Result~\ref{res:2} generalizes the result of the case where $n$ agents have identical ordering cost functions in~\cite{li2022almost}.
Informally speaking, an identical ordering means that the agents have the same ordinal preference for the chores.
See Section~\ref{sec:n-1} for the formal definition of identical ordering.
It should be emphasized that, in Result~\ref{res:2}, the remaining agent can have a general cost function.
Note that our result also extends the result in~\cite{yin2022envy}, in which they considered a more restricted case as mentioned above.

In the last result, we consider personalized bi-valued instances, in which 
each agent has two values for chores that may be different.
Thus, personalized bi-valued instances include the bi-valued instances but not vice versa.
Result~\ref{res:3} extends the result in \cite{zhou2021approximately}, where a positive result was shown when $n=3$ and each agent has a bi-valued cost function.
\subsection{Related Work}

\paragraph{Envy-freeness up to {\it one} item (EF1) for goods}
One of the most popular relaxed notions of EF is EF1, which was introduced by Budish~\cite{budish2011combinatorial}.
EF1 requires that each agent $i$ prefers their own bundle to be at least as good as the bundle of agent $j$ after the removal of {\it some} good in $j$'s bundle.
Thus, EF1 is a weaker notion than EFX.
While the existence of EFX allocations remains open in general, an EF1 allocation can be computed in polynomial time for any instance~\cite{lipton2004approximately, caragiannis2019unreasonable}.
There are several studies that find not only EF1, but also efficient (particularly Pareto optimal) allocation.
It is known that the maximum Nash social welfare solution satisfies both EF1 and PO (Pareto optimal)\cite{caragiannis2019unreasonable}.
Barman et al.~\cite{barman2018finding} show that an allocation satisfying both EF1 and PO can be computed in pseudo-polynomial time.
It remains an open problem whether a polynomial-time algorithm exists to find an allocation that satisfies both EF1 and PO.

\paragraph{Envy-freeness up to {\it one} item (EF1) for chores}
By analogy to goods, EF1 for chores can be defined as follows.
An allocation $A=(A_1, A_2, \ldots , A_n)$ of chores is called {\it EF1} if for all pairs of agents $i$ and $j$ with $A_i\neq \emptyset$, and for some $e \in A_i$, it holds that $c_i(A_i\setminus \{e\}) \le c_i(A_j)$.
In other words, each agent $i$ prefers their own bundle to be at least as good as the bundle of agent $j$ after the removal of {\it some} chore in $i$'s bundle.
Bhaskar et al. \cite{bhaskar2021approximate} showed an EF1 allocation of chores always exists and can be computed in polynomial time.
It remains open whether there always exists an allocation that satisfies both EF1 and PO for chores.

\paragraph{Approximate EFX Allocations}
There are several studies on approximate EFX allocations.
In the case of goods, the definition of $\alpha$-EFX is obtained by replacing $v_i(A_i) \ge v_i(A_j\setminus \{g\})$ with $v_i(A_i) \ge \alpha \cdot v_i(A_j\setminus \{g\})$ in the definition of EFX, where $\alpha\in [0,1]$.
It is known that there are $1/2$-EFX allocations for subadditive valuations\cite{plaut2020almost}.
For additive valuations, there are polynomial time algorithms to compute $0.618$-EFX allocations\cite{amanatidis2020multiple, farhadi2021almost}.
As for chores, there has been little research done so far.
In the case of chores, the definition of $\alpha$-EFX is obtained by replacing $c_i(A_i\setminus \{e\}) \le c_i(A_j)$ with $c_i(A_i\setminus \{e\}) \le \alpha \cdot c_i(A_j)$ in the definition of EFX, where $\alpha\ge 1$.
Zhou and Wu\cite{zhou2021approximately} showed that there exists a polynomial time algorithm to compute a $5$-EFX allocation for $3$ agents and a $3n^2$-EFX allocation for $n\ge 4$ agents.
It remains open whether there exist constant approximations of EFX allocation for any number of agents.

\subsection{Organization}
Section~\ref{sec:pre} provides definitions for terminology and notations, defines the EFX-graph, and discusses its basic properties.
In Section~\ref{sec:2n}, it is shown that an EFX allocation exists when the number of chores is at most twice the number of agents.
In Section~\ref{sec:n-1}, we consider the case where the cost functions of all agents except for one are identical ordering.
In Section~\ref{sec:bi-valued}, we consider the case where $n=3$ and each agent has a personalized bi-valued cost function.
\section{Preliminaries}\label{sec:pre}

Let $N=\{1, 2, \ldots, n\}$ be a set of $n$ agents and $M$ be a set of $m$ indivisible chores.
Each agent $i\in N$ has a cost function $c_i: 2^M \rightarrow \mathbb{R}_{\ge 0}$.
We assume that (i) any cost function $c_i$ is {\it normalized}: $c_i(\emptyset) = 0$, (ii) it is {\it monotone}:  $S\subseteq T$ implies $c_i(S) \le c_i(T)$ for any $S,T \subseteq M$, and (iii) it is {\it additive}: $c_i(S)=\sum_{e\in S} c_i({e})$ for any $S\subseteq M$.

To simplify notation, we denote $\{1,\dots,k\}$ by $[k]$.
For any $i\in N$, $e\in M$, and $S\subseteq M$, we write $c_i(e)$ to denote $c_i(\{e\})$, and use $S\setminus e, S\cup e$ to denote $S\setminus \{e\}, S\cup \{e\}$, respectively.

For $M' \subseteq M$, an {\it allocation} $A=(A_1,A_2,\ldots, A_n)$ of $M'$ is an $n$-partition of $M'$, where $A_i\cap A_j = \emptyset$ for all $i$ and $j$ with $i \neq j$, and $\bigcup_{i\in N} A_i= M'$.
In an allocation $A$, $A_i$ is called a {\it bundle} given to agent $i\in N$. 
Given an allocation $A$, we say that agent $i$ {\it envies} agent $j$ if $c_i(A_i)> c_i(A_j)$, and agent $i$ {\it strongly envies} agent $j$ if there exists a chore $e$ in $A_i$ such that $c_i(A_i\setminus e) > c_i(A_j)$.
An allocation $A$ is called {\it EFX} if no agent strongly envies another, i.e., for any pair of agents $i, j\in N$ and $e\in A_i$, $c_i(A_i\setminus e)\le c_i(A_j)$.
It is easy to see that an allocation $A$ is EFX if and only if for any agent $i\in N$, we have $\max_{e\in A_i} c_i(A_i\setminus e) \le \min_{j\in [n]} c_i(A_j)$. 

Let $G=(V, E)$ be a graph.
A {\it matching} in $G$ is a set of pairwise disjoint edges of $G$.
A {\it perfect matching} in $G$ is a matching covering all the vertices of $G$.
For a matching $X$ in $G$ and a vertex $v$ incident to an edge in $X$, we write $X(v)$ as the vertex adjacent to $v$ in $X$.
For a subgraph $H$ of $G$, let $V[H]$ denote all vertices in $H$ and $E[H]$ denote all edges in $H$.
Similarly, $V[H]$ and $E[H]$ are defined also for a digraph $H$. 
For finite sets $A$ and $B$, we denote the symmetric difference of $A$ and $B$ as $A\bigtriangleup B=(A\setminus B)\cup (B\setminus A)$.

\paragraph{EFX-graph}
In this paper, we use a bipartite graph called an {\it EFX-graph}, which plays an important role to show Results~\ref{res:1} to \ref{res:3}. 
We now define the EFX-graph and provide its basic properties.

Let $U$ be a set of size $n$ and $M'\subseteq M$.
We say that $A=(A_u)_{u\in U}$ is an {\it allocation to $U$ of $M'$} if it is an $n$-partition of $M'$, where each set is indexed by an element in $U$, i.e., 
$A_u \cap A_{u'} = \emptyset$ for all $u$ and $u'$ with $u\neq u'$, and $\bigcup_{u \in U} A_u = M'$. 
For an allocation $A=(A_u)_{u\in U}$ to $U$, we define a bipartite graph $G_{A}=(N, U; E_A)$ called {\it EFX-graph} as follows.
The vertex set consists of $N$ and $U$, and the edge set $E_A$ is defined by
$$
(i,u) \in E_A \iff \max_{e\in A_u} c_i(A_u\setminus e) \le \min_{k\in U} c_i(A_k)
$$
for any $i\in N$ and $u\in U$.
That is, an edge $(i, u)$ means that agent $i$ can receive $A_u$ without violating the EFX conditions.
We define $E_A^{\min}$ by the set of all edges corresponding to the minimum cost chore set in $A$, i.e., $E_A^{\min}=\{(i,u)\in E_A\mid c_i(A_u)=\min_{k\in U} c_i(A_k), i\in N \}$.
See also Example~\ref{ex:1}.
%We call $E_A^{\min}$ {\it the minimum edge set}.
\begin{example}\label{ex:1}
Let $N=\{1,2,3\}$ and $M=\{e_1, e_2, e_3, e_4, e_5, e_6\}$, and suppose that each cost function is represented as in Table~\ref{tab:examle}.
Consider an allocation $A=(\{e_1, e_2\}, \{e_3, e_4\}, \{e_5, e_6\})$ to $U$.
Then, the corresponding EFX-graph $G_A$ is represented as in Figure~\ref{fig:examle1}.
\end{example}
\begin{figure}[bhtp]
 \begin{tabular}{cc}
 \begin{minipage}{0.5\hsize}
  \begin{center}
  \large
    \begin{tabular}{c|cccccc}
         & $e_1$ & $e_2$ & $e_3$ & $e_4$ & $e_5$ & $e_6$  \\ \hline
         agent 1& 2 & 0 & 5 & 2 & 5 & 2 \\
         agent 2& 2 & 4 & 3 & 3 & 0 & 3 \\
         agent 3& 1 & 1 & 1 & 1 & 1 & 1 \\
         
    \end{tabular}
    \tblcaption{The cost of each agent's chores}
    \label{tab:examle}
   \end{center}
 \end{minipage}
 \hfill
 \begin{minipage}{0.5\hsize}
   \begin{center}
   \includegraphics[width=35mm]{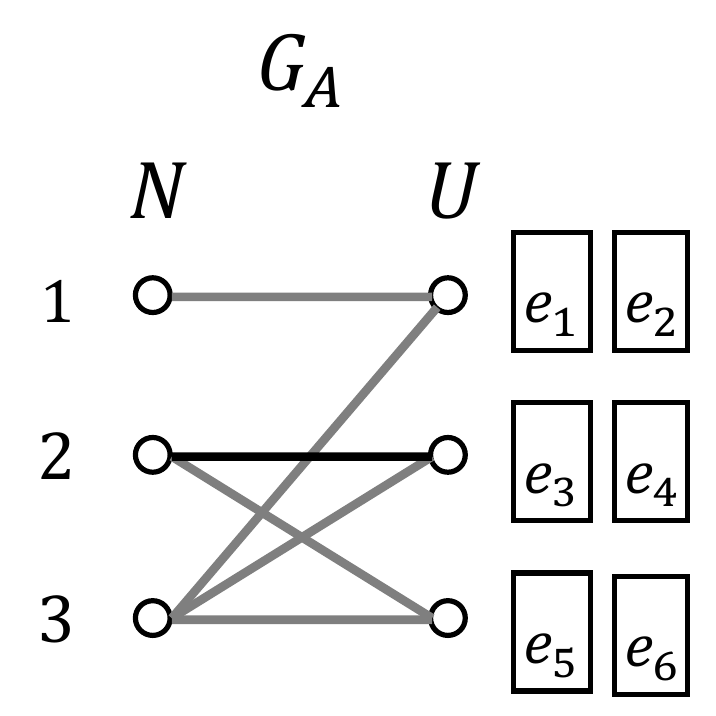}
   \caption{EFX-graph $G_A$. The black and gray edges are in $E_A$, and the gray edges are in $E^{\min}_A$.}
     \label{fig:examle1}
   \end{center}
 \end{minipage}\\
 \end{tabular}
\end{figure}

By simple observation, we see the following properties hold.
\begin{observation}\label{ob:simple}
Let $G_{A}=(N, U; E_A)$ be an EFX-graph for an allocation $A=(A_u)_{u\in U}$.
Then, the following properties hold.
\begin{enumerate}
    \item[(i)] For $u\in U$ with $|A_u|\le 1$, it holds that $(i, u)\in E_A$ for any $i \in N$.
    \item[(ii)] For any $i \in N$, there exists $u \in U$ such that $(i, u)\in E_A^{\min}$.
    \item[(iii)] If $G_A$ has a perfect matching, then $G_A$ has a perfect matching $X$ such that $(i, X(i))\in E^{\min}_{A}$ for some $i\in N$.
\end{enumerate}
\end{observation} 
\begin{proof}
    (i) If $|A_u| =0$, the claim is obvious. If $|A_u| = 1$, then $\max_{e\in A_u} c_i(A_u\setminus e) = 0 \le \min_{k\in U} c_i(A_k)$ for any $i\in N$. Thus, we have $(i,u) \in E_A$.

    (ii) For $i\in N$, let $A_u\in \argmin_{k\in U} c_i(A_k)$. Then, we have $(i,u)\in E^{\min}_{A}$.

    (iii) Let $Y$ be any perfect matching in $G_A$. If $(i, Y(i))\in E^{\min}_{A}$ for some $i \in N$, then we are done.
    Suppose that $(i, Y(i))\not\in E^{\min}_{A}$ for any $i\in N$.
    We consider a directed bipartite graph $D_A=(N, U; F)$, where the vertex set consists of $N$ and $U$, and the arc set $F$ is defined by
$$
F=\{(i,u)\mid (i,u)\in E^{\min}_A, i\in N\}\cup \{(u, Y(u))\mid u\in U\}.
$$
Since all vertices in $D_A$ have at least one outgoing arc by (ii), $D_A$ has a directed cycle $\overrightarrow{C}$ of length more than two by our assumption.
Let $C$ be the underlying undirected cycle of $\overrightarrow{C}$.
We define a new perfect matching $X=Y\bigtriangleup E[C]$ in $G_A$.
Then, there exists an edge $(i, X(i))\in E^{\min}_{A}$ for some $i\in N$. 
\end{proof}

It is easy to see that if an EFX-graph has a perfect matching, then an EFX allocation can be obtained as follows.
\begin{observation}\label{ob:final}
Let $M'\subseteq M$ and $A=(A_u)_{u\in U}$ be an allocation to $U$ of $M'$.
Suppose that EFX-graph $G_A$ has a perfect matching.
Then, there exists an EFX allocation of $M'$, and it can be found in polynomial time.
\end{observation} 
\begin{proof}
Let $X$ be a perfect matching in $G_A$.
We construct an allocation $A'=(A'_1,\ldots, A'_n)$ to $N$ as follows.
For each agent $i\in N$, $A'_i$ is defined as $A_{X(i)}$, where $A_{X(i)}$ is the chore set corresponding to the vertex matched to $i$ in $X$.
By the definition of $E_A$, it holds that $\max_{e\in A'_i} c_i(A'_i\setminus e)=\max_{e\in A_{X(i)}} c_i(A_{X(i)}\setminus e) \le \min_{k\in U} c_i(A_k)\le \min_{j\in [n]} c_i(A'_j)$ for any $i\in N$. 
Thus, $A'$ is an EFX allocation of $M'$. 
\end{proof}
Let $M'\subsetneq M$ and $A=(A_u)_{u\in U}$ be an allocation to $U$ of $M'$.
Let $e\in M\setminus M'$ be an unallocated chore and $A_v$ be some chore set in $A$.
We say that an allocation $A'=(A'_u)_{u\in U}$ to $U$ of $M'\cup e$ is obtained from $A$ by {\it adding} $e$ to $A_v$ if
\begin{align*}
A'_u &= \left\{ 
\begin{array}{ll}
A_{v}\cup e & {\rm if}~u=v,\\
A_{u} & {\rm otherwise}.
\end{array}
 \right. &
\end{align*}
The following lemma is a fundamental one that will be used repeatedly later.
\begin{lemma}\label{lem:basic}
Let $M'\subsetneq M$ and $A=(A_u)_{u\in U}$ be an allocation to $U$ of $M'$.
Suppose that there exist $i \in N$ and $e\in M\setminus M'$ such that $c_i(e)\le c_i(e')$ for any $e'\in M'$.
Let $(i, u_i)\in E_A^{\min}$ and $A'=(A'_u)_{u\in U}$ be an allocation obtained from $A$ by adding $e$ to $A_{u_i}$.
Then, the following two statements hold.
\begin{enumerate}
\item[(i)] $(i, u_i)\in E_{A'}$. 

\item[(ii)] $(j, u)\in E_A \Rightarrow (j, u)\in E_{A'}$ for any $j\in N$ and $u\in U\setminus u_i$.
\end{enumerate}
\end{lemma}
\begin{proof}
(i) Since $c_i(e)\le c_i(e')$ for any $e'\in M'$, $\max_{f\in A'_{u_i}} c_i(A'_{u_i} \setminus f)=c_i(A'_{u_i}\setminus e)=c_i(A_{u_i})$.
By the fact that $(i, u_i)\in E_A^{\min}$, $c_i(A_{u_i}) = \min_{k\in U} c_i(A_k) \le  \min_{k\in U} c_i(A'_k)$.
Therefore, $\max_{f\in A'_{u_i}} c_i(A'_{u_i}\setminus f)\le \min_{k\in U} c_i(A'_k)$, which implies $(i, u_i)\in E_{A'}$.

(ii) Fix any $j\in N $ and $u\in U\setminus u_i$ with $(j, u)\in E_A$.
Since $u\in U\setminus u_i$, we have $A'_{u}=A_{u}$.
Thus, $\max_{f\in A'_u} c_j(A'_u\setminus f)=\max_{f\in A_u} c_j(A_u\setminus f) \le \min_{k\in U} c_j(A_k)\le \min_{k\in U} c_j(A'_k)$, where the first inequality follows from $(j,u)\in E_A$.
This implies $(j, u)\in E_{A'}$.
\end{proof}
The following corollary can be obtained by applying Lemma~\ref{lem:basic} for $i$, $e$, and $u_i=X(i)$.
\begin{corollary}\label{cor:1}
Let $M'\subsetneq M$ and $A=(A_u)_{u\in U}$ be an allocation to $U$ of $M'$.
Suppose that $G_A$ has a perfect matching $X$ such that $(i, X(i)) \in E^{\min}_{A}$, and there exists $e\in M\setminus M'$ such that $c_i(e)\le c_i(e')$ for any $e'\in M'$.
Then, $X$ is a perfect matching also in $G_{A'}$, where $A'=(A'_u)_{u \in U}$ is the allocation obtained from $A$ by adding $e$ to $A_{X(i)}$.
%%Then, $G_{A'}$ has a perfect matching, where $A'=(A'_u)_{u \in U}$ is an allocation obtained from $A$ by adding $e$ to $A_{X(i)}$.
\end{corollary}
\section{Existence of EFX with at most $2n$ Chores}\label{sec:2n}
In this section, we prove the following theorem by constructing a polynomial-time algorithm to find an EFX allocation when $m\le 2n$.
\begin{theorem}\label{thm:2n}
There exists an EFX allocation of chores when $m\le 2n$ and each agent has an additive cost function.
Moreover, we can find an EFX allocation in polynomial time.
\end{theorem}
Our algorithm is described in Algorithm~\ref{alg:01}.
If $m\le n$, then an EFX allocation can be obtained by allocating at most one chore to each agent. Otherwise, we denote $m=n+l$ with $1\leq l\leq n$. 
Our basic idea is as follows.
First, we create an allocation $A=(A_u)_{u\in U}$ to $U$ by setting $A_u=\emptyset$ for any $u\in U$. Then, we add chores to one of the chore sets in $A$ one by one while maintaining the condition that $G_A$ has a perfect matching. 
If this condition is satisfied after all chores have been allocated, we can obtain an EFX allocation of $M$ by Observation~\ref{ob:final}.
However, in general, it is not possible to maintain this condition when adding chores in an arbitrary order. Intuitively, this is because it becomes difficult to keep an edge in $G_A$ if heavy chores are added at the end.
To address this issue, we first let $l$ agents (e.g., agents $l, l-1,\dots,1$) choose the chore with the smallest cost for themselves in turn and hold it. 
Then, we create an allocation $A=(A_u)_{u\in U}$ for the remaining $n$ chores, such that $|A_u|=1$ for any $u\in U$. 
Next, we add the held chores to the smallest chore set for each agent in the reverse order ($1, 2, \dots, l$). By applying induction, 
we can show that $G_A$ always has a perfect matching during the entire process of adding chores.
%%it can be shown that the condition that $G_A$ has a perfect matching is maintained during the entire process of adding chores.

\begin{algorithm}[htb]
%%\caption{ Algorithm to find an EFX allocation when $m\le 2n$}
\caption{ Case when $m\le 2n$}
\label{alg:01}
\begin{algorithmic}[1]
\Require a set of agents $N$, a set of chores $M$ of size at most $2n$, and a cost function $c_i$ for each $i \in N$
\Ensure  an EFX allocation $A^*$ of $M$.
\State $R\leftarrow M$ 
\Comment{remaining set of chores}
\State $A_u \leftarrow \emptyset$ for all $u\in U$, where $U$ is a set of size $n$.
\State Set $l=\max\{m-n, 0\}$.
\For {$i=l, l-1, \cdots, 1$}
    \State Pick up $e_i \in \argmin_{e\in R} c_i(e)$
    \State $R\leftarrow R\setminus e_i$
\EndFor
\For {$u\in U$}
    \State Pick up $e\in R$ arbitrarily (if it exists).
    \State $A_{u}\leftarrow \{e\}, R\leftarrow R\setminus e$
\EndFor
\For {$i=1, 2, \cdots, l$}
    \State Pick up $u_i\in \argmin_{u\in U} c_i(A_u)$
    \State $A_{u_i}\leftarrow A_{u_i}\cup e_i$
\EndFor
\State Find a perfect matching $X$ on $G_A$.
\State Construct the allocation $A^*$ by allocating each chore set to the  matched agent in $X$.\\
\Return $A^*$
\end{algorithmic}
\end{algorithm}
\begin{figure}[tbp]
    \centering
    \includegraphics[width=100mm]{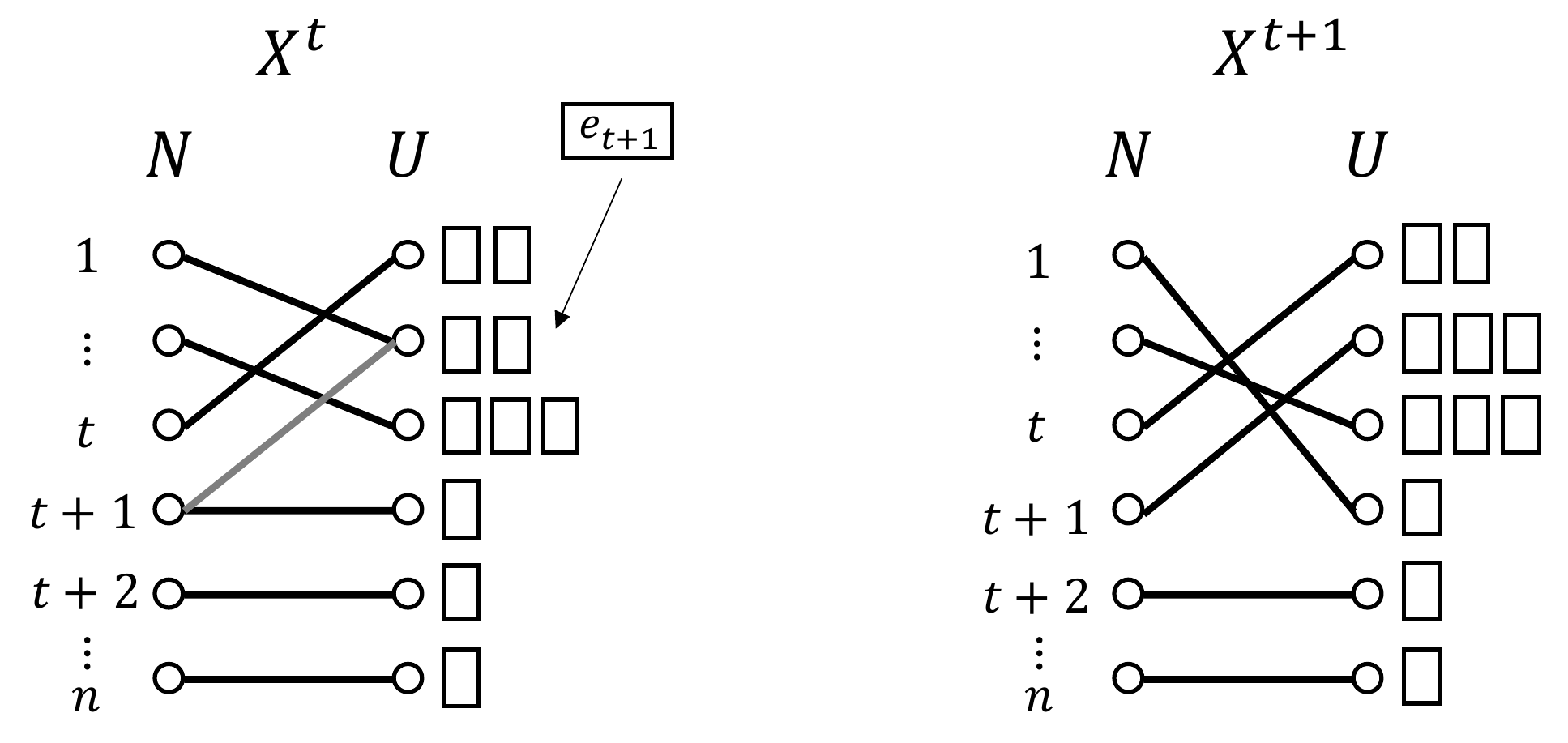}
    \caption{Situation in Case 2 of Claim~\ref{cl:induction}.
    The black edge set in the left figure represents $X^t$, and the black edge set in the right figure represents $X^{t+1}$.
    The gray edge in the left figure represents an edge in $E^{\min}_{A^t}$.
    }
    \label{fig:claim}
\end{figure}
\begin{proof}[Proof of Theorem~\ref{thm:2n}]
We first show the correctness of Algorithm~\ref{alg:01}.
By Observation~\ref{ob:simple}, it is sufficient to show that $G_A$ has a perfect matching in line $13$ of Algorithm~\ref{alg:01}.
If $l=0$, that is, $m\le n$, the first and third for-loops are not executed.
After the middle for-loop is executed, $|A_u|\le 1$ for any $u\in U$ and $R=\emptyset$.
Hence, $A=(A_u)_{u\in U}$ is an allocation of $M$ to $U$ and $G_A$ becomes a perfect bipartite graph by Observation~\ref{ob:simple} (i).
Thus, $G_A$ has a perfect matching.

Suppose next that $l > 0$. 
For $t\in \{0\}\cup [l]$, we call the $t$-th execution of the third for-loop  {\it round $t$}.
Let $A^t=(A^t_u)_{u \in U}$ be the allocation to $U$ immediately after adding a chore in round $t$.
We show the existence of a perfect matching in $G_{A^t}$ by induction on $t$.
More precisely, we show the following claim.
\begin{claim}\label{cl:induction}
For any round $t\in \{0\}\cup [l]$, there exists a perfect matching $X^t$ in $G_{A^t}$ such that 
\begin{equation}
\label{cond:01}
(i,u)\in X^t \Rightarrow |A^t_u|=1 {\it \ for\ any\ } i\in \{t+1,\ldots, n\}.
\end{equation}
%%|A^t_{X^t(i)}|=1 {\it \ for\ any\ } i\in \{t+1,\ldots, n\}.
%%\mbox{for any $i\in \{t+1,\ldots, n\}$, $(i,u)\in X^t$ implies $|A^t_u|=1$.}
\end{claim}
\begin{proof}[Proof of Claim~\ref{cl:induction}]
We show the claim by induction on $t$.
For the base case $t=0$, after the middle for-loop is executed, $|A_u|= 1$ holds for any $u\in U$.
Hence, $G_A$ has a perfect matching and condition (\ref{cond:01}) obviously holds.
For the inductive step, we assume that Claim~\ref{cl:induction} holds for $t\in \{0\}\cup [l]$.
Let $X^t$ be a perfect matching in $G_{A^t}$ satisfying (\ref{cond:01}).
Since the first and third for-loops are executed in the reverse order with respect to $i$, 
it holds that 
$$
c_{t+1}(e_{t+1}) \le c_{t+1}(e') {\rm \ for \ any\ } e'\in \bigcup_{u\in U} A^t_u.
$$
Thus, by applying Lemma~\ref{lem:basic} for $i=t+1$ and $e=e_{t+1}$, we obtain the following, 
where we recall that $u_{t+1} \in \argmin_{u\in U} c_{t+1}(A^t_u)$. 
\begin{enumerate}
\item[(i)] $(t+1, u_{t+1})\in E_{A^{t+1}}$. 
\item[(ii)] $(j, u)\in E_{A^t} \Rightarrow (j, u)\in E_{A^{t+1}}$ for any $j\in N$ and $u\in U\setminus u_{t+1}$.
\end{enumerate}
We consider two cases separately.
\begin{description}
\item [Case 1:]$(t+1, u_{t+1})\in X^t$\\
Let $X^{t+1}=X^t$.
By (i) and (ii) above, we obtain $X^{t+1}\subseteq E_{A^{t+1}}$.
In addition, for $i\in \{t+2,\ldots, n\}$, 
if $(i,u)\in X^{t+1}=X^t$, then $|A^{t+1}_u|=|A^{t}_u|=1$ holds by the induction hypothesis.
Thus, $X^{t+1}$ is a perfect matching in $E_{A^{t+1}}$ satisfying condition (\ref{cond:01}).

\item [Case 2:]$(t+1, u_{t+1})\notin X^t$ 

In this case, we create a new perfect matching by swapping two edges in $X^t$ (see Figure~\ref{fig:claim}).
Formally, we define $X^{t+1}=X^t\cup \{(t+1, u_{t+1}), (X^t(u_{t+1}), X^t(t+1))\}\setminus \{(t+1, X^t(t+1)), (X^t(u_{t+1}), u_{t+1})\}$.
We see that $(t+1, u_{t+1})\in E_{A^{t+1}}$ holds by (i) above and 
$(X^t(u_{t+1}), X^t(t+1)) \in E_{A^{t+1}}$ holds by Observation~\ref{ob:simple} (i).
%%By the facts (i), (ii), and $X^{t}(t+1)\neq u_{t+1}$, it holds that $(t+1, u_{t+1}), (X^t(u_{t+1}), X^t(t+1)) \in E_{A^{t+1}}$.
Hence, $X^{t+1}\subseteq E_{A^{t+1}}$.
For any $i\in \{t+2,\ldots, n\}$, if $i=X^{t}(u_{t+1})$, then $(i, X^t(t+1))\in X^{t+1}$, and $|A^{t+1}_{X^t(t+1)}|=|A^t_{X^t(t+1)}|=1$ by the induction hypothesis and $X^{t}(t+1)\neq u_{t+1}$.
Otherwise, since $A^{t+1}_{X^{t+1}(i)}=A^{t}_{X^t(i)}$, $|A^{t+1}_{X^{t+1}(i)}|=|A^{t}_{X^t(i)}|=1$ by the induction hypothesis.
Thus, $X^{t+1}$ is a perfect matching in $E_{A^{t+1}}$ satisfying condition (\ref{cond:01}). 
\end{description}
\end{proof}
By Claim~\ref{cl:induction}, there exists a perfect matching $X^{l}$ in $G_{A^{l}}$.
This means that there exists a perfect matching $X$ in $G_{A}$ in line $13$ of Algorithm~\ref{alg:01}.
Therefore, Algorithm~\ref{alg:01} returns an EFX allocation.

We next show that Algorithm~\ref{alg:01} runs in polynomial time.
For the first for-loop, line 5 is executed in $O(m)$ time for each $i$.
Since $l\le n$, the first for-loop takes $O(mn)$ to execute.
The second for-loop takes $O(n)$ to execute since $|U|\le n$.
The last for-loop takes $O(mn)$ to execute.
Finally, we can find a perfect matching $X$ on $G_A$ by a maximum matching algorithm on a bipartite graph in $O(mn^2)$ time, 
because $G_A$ has $O(n)$ vertices and $O(mn)$ edges.
Therefore, Algorithm~\ref{alg:01} returns an EFX allocation in $O(mn^2)$ time. 
Note that the running time can be improved by using a sophisticated bipartite matching algorithm, but we do not go into details. 
\end{proof}

\section{When $n-1$ Agents Have Identical Ordering Cost Functions}\label{sec:n-1}
In this section, we consider the case where $n-1$ agents have identical ordering cost functions.
For any pair of agents $i$ and $j$, we call the cost functions of $i$ and $j$ are {\it identical ordering} if for any $e$ and $e'$ in $M$, it holds that $c_i(e)<c_i(e') \iff c_j(e)<c_j(e')$.
In other words, we can sort all the chores in non-increasing order of cost for both $i$ and $j$.
For $k\in [n]$, we say that $k$ agents have identical ordering cost functions if the cost functions of any two agents among those $k$ agents are identical ordering.
We show the following theorem.
\begin{theorem}\label{thm:ordering}
There exists an EFX allocation of chores when $n-1$ agents have identical ordering cost functions.
Moreover, we can find an EFX allocation in polynomial time.
\end{theorem}
From now on, we assume that agents $1,2,\ldots n-1$ have identical ordering cost functions.
Our algorithm is described in Algorithm~\ref{alg:02}.
Our basic idea is quite similar to the approach in Section~\ref{sec:2n}.
First, we create an allocation $A=(A_u)_{u\in U}$ to $U$ by setting $A_u=\emptyset$ for any $u\in U$.
We sort the chores in non-increasing order of cost for agents $1,2,\ldots, n-1$, which is possible as they have identical ordering cost functions. 
Then, in this order, we add chores to one of the chore sets in $A$ one by one while maintaining the condition that $G_A$ has a perfect matching. 
If this condition is satisfied after all chores have been allocated, we can obtain an EFX allocation of $M$ by Observation~\ref{ob:final}.
%%As in Section~\ref{sec:2n}, 
By applying induction, we can show that $G_A$ always has a perfect matching during the entire process of adding chores.
\begin{algorithm}[htb]
\caption{ Case when $n-1$ agents have identical ordering cost functions.}
\label{alg:02}
\begin{algorithmic}[1]
\Require a set of agents $N$, a set of chores $M$, and a cost function $c_i$ for each $i \in N$, where $c_1, \ldots , c_{n-1}$ are identical ordering.
\Ensure  an EFX allocation $A^*$ of $M$.
\State $A_u \leftarrow \emptyset$ for all $u\in U$, where $U$ is a set of size $n$.
\State Sort all the chores in $M$: $c_i(e_1)\ge c_i(e_2)\ge \cdots \ge c_i(e_m)$ for all $i\in [n-1]$.
\For {$t=1, 2, \cdots, m$}
    \State Find a vertex $u^*\in U$ such that $G_{A'}$ has a perfect matching, 
    \State where $A'$ is the allocation to $U$ obtained from $A$ by adding $e_t$ to $A_{u^*}$.
    \State $A_{u^*}\leftarrow A_{u^*}\cup e_t$ 
\EndFor
\State Find a perfect matching $X$ on $G_A$.
\State Construct the allocation $A^*$ by allocating each chore set to the matched agent in $X$.\\
\Return $A^*$
\end{algorithmic}
\end{algorithm}

\begin{figure}[tbp]
    \centering
    \includegraphics[width=100mm]{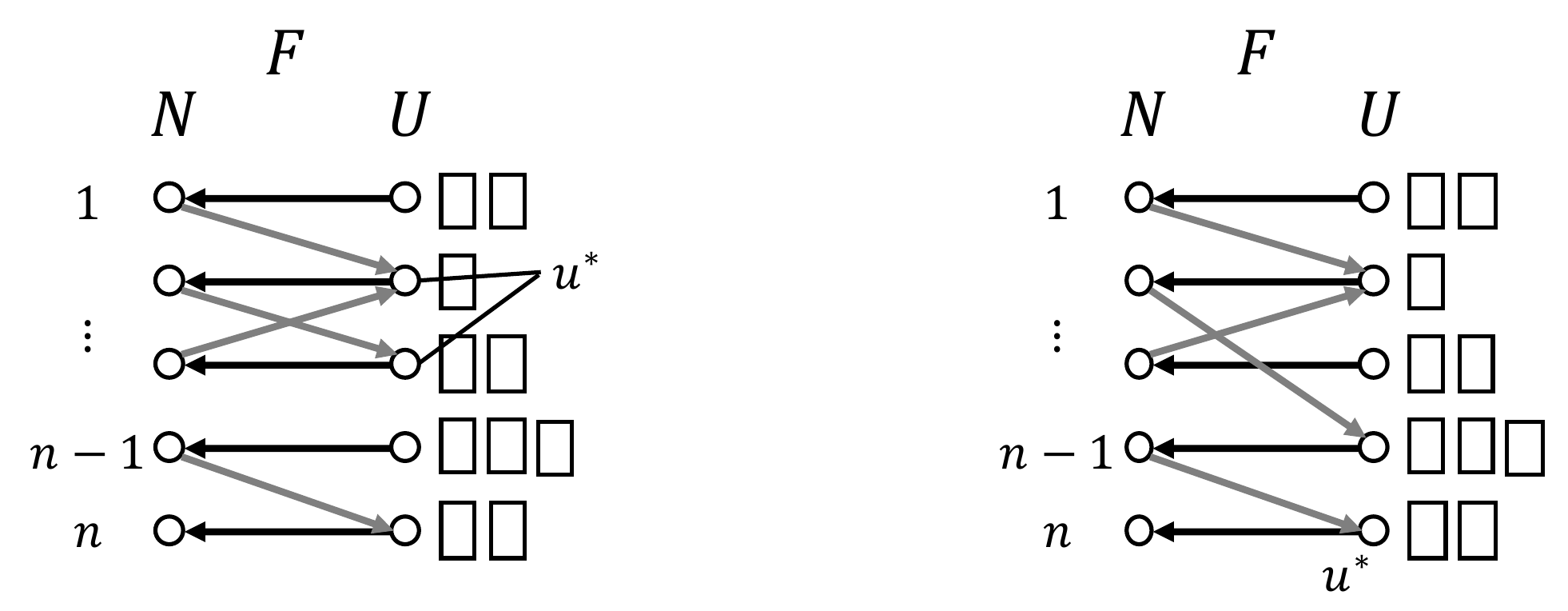}
    \caption{Two cases in Lemma~\ref{cl:add}: 
    $D_A$ has a directed cycle (left) and $D_A$ has no directed cycles (right).
    }
    \label{fig:claim2}
\end{figure}

In order to show that there exists a vertex $u^*$ satisfying the desired condition in lines 4 and 5 of Algorithm~\ref{alg:02}, we prove the following lemma, which will be used also in Section~\ref{sec:bi-valued}.

\begin{lemma}\label{cl:add}
Let $M'\subsetneq M$ and $e \in M \setminus M'$ such that $c_i(e)\le c_i(e')$ for any $i\in [n-1]$ and $e'\in M'$.
Let $A=(A_u)_{u\in U}$ be an allocation to $U$ of $M'$ such that $G_{A}$ has a perfect matching. 
Then, there exists a vertex $u^*\in U$ such that $G_{A'}$ has a perfect matching, where $A'$ is the allocation to $U$ of $M'\cup e$ obtained from $A$ by adding $e$ to $A_{u^*}$.
\end{lemma}

\begin{proof}
Let $X$ be a perfect matching on $G_{A}$, and $e\in M\setminus M'$ be an unallocated chore such that 
$c_i(e)\le c_i(e')$ for any $i\in [n-1]$ and $e'\in M'$.
For an allocation $A=(A_u)_{u\in U}$ to $U$ of $M'$, 
we consider a directed bipartite graph $D_A=(N, U; F)$, where the vertex set consists of $N$ and $U$, and the arc set $F$ is defined by
$$
F=\{(i,u)\mid (i,u)\in E^{\min}_A, i\in [n-1]\}\cup \{(u, X(u))\mid u\in U\}.
$$
We consider two cases separately. See also Figure~\ref{fig:claim2}.
\begin{description}
\item[Case 1:] $D_A$ has a directed cycle.\\
Let $\overrightarrow{C}$ be any directed cycle in $D_A$.
Since the vertex $n\in N$ has no outgoing arc in $D_A$, we have $n\not \in V[\overrightarrow{C}]$.
If $|V[\overrightarrow{C}]|=2$, define a perfect matching $X'=X$ in $G_A$.
Otherwise, define a new perfect matching $X'=X\bigtriangleup E[C]$ in $G_A$, where $C$ is the underlying (undirected) cycle of $\overrightarrow{C}$. 
In both cases, we pick up $i\in V[\overrightarrow{C}]\cap [n-1]$ and choose $X'(i)$ as $u^*$. 
Then, by the definition of $F$, it holds that $(i, X'(i))\in E^{\min}_A$.
By applying Corollaly~\ref{cor:1}, $G_{A'}$ has a perfect matching $X'$, where $A'$ is the allocation to $U$ of $M'\cup e$ obtained from $A$ by adding $e$ to $A_{X'(i)}$.
\item[Case 2:] $D_A$ has no directed cycles.\\
We choose $X(n)$ as $u^*$.
Let $A'$ be an allocation to $U$ of $M'\cup e$ obtained from $A$ by adding $e$ to $A_{X(n)}$.
By the same argument as in Lemma~\ref{lem:basic} (ii), we can see that
\begin{equation}
\label{cond:02}
(j, u)\in E_{A} \Rightarrow (j, u)\in E_{A'} \ {\rm for\ any}\ j\in N \ {\rm and}\  u\in U\setminus X(n).
\end{equation}
If $(n, X(n))\in E_{A'}$, then $X'=X$ is a perfect matching in $G_{A'}$ by (\ref{cond:02}).
Otherwise, there exists $u\neq X(n)$ with $(n,u)\in E^{\min}_{A'}$ by Observation~\ref{ob:simple} (ii).
Since all vertices except $n$ have at least one outgoing arc in the acyclic digraph $D_A$, all vertices have a directed path to $n$ in $D_A$.
Let $\overrightarrow{P}$ be a directed path from $u$ to $n$ in $D_{A}$ and 
$P$ be the underlying path of $\overrightarrow{P}$.
Note that $P \cup (n, u)$ forms a cycle, and it traverses edges in $X$ and ones not in $X$ alternately by the definition of $F$. 
This shows that $X':= X \bigtriangleup (P\cup (n, u))$ is a perfect matching in $E_A \cup (n, u)$ such that $(X'(u^*), u^*) \in E^{\min}_{A}$. 
Since each edge in $X' \setminus \{(n, u), (X'(u^*), u^*)\}$ is in $E_{A'}$ by (\ref{cond:02}), $(X'(u^*), u^*)$ is in $E_{A'}$ by Lemma~\ref{lem:basic} (i), and $(n,u)\in E^{\min}_{A'}$,  
$X'$ is a perfect matching in $G_{A'}$. 
\end{description}
\end{proof}

We are now ready to prove Theorem~\ref{thm:ordering}. 

\begin{proof}[Proof of Theorem~\ref{thm:ordering}]
We first show the correctness of Algortihm~\ref{alg:02}.
By Lemma~\ref{cl:add}, we can pick up $u^*$ satisfying the desired condition in lines 4 and 5 of Algorithm~\ref{alg:02}.
During the execution of the for-loop, we maintain the condition that $G_A$ has a perfect matching.
Thus, there exists a perfect matching in line 7 of Algorithm~\ref{alg:02}.
Therefore, we can find an EFX allocation by Observation~\ref{ob:final}.

We next show that Algorithm~\ref{alg:02} runs in polynomial time.
Line 2 is easily done in polynomial time using a sorting algorithm.
Lines 4 and 5 can be executed in polynomial time since we can check the condition by using a maximum matching algorithm for all $u^*\in U$, which can be done in polynomial time.
Thus, the for-loop runs in polynomial time.
By applying a maximum matching algorithm again, we can find a perfect matching in line 7 in polynomial time.
Therefore, Algorithm~\ref{alg:02} runs in polynomial time.
\end{proof}

We give a remark here that, in our proofs of Lemma~\ref{cl:add} and Theorem~\ref{thm:ordering}, we do not use the explicit form of the cost function of the remaining agent (agent $n$). Therefore, we can slightly generalize Theorem~\ref{thm:ordering} so that the remaining agent can have a general (i.e., non-additive) cost function.

\section{Personalized Bi-valued Instances}\label{sec:bi-valued}
In this section, we consider the case where the number of agents is three and each agent has a {\it personalized bi-valued} cost function, which means that, for any $i\in N (=[3])$,  
there exist $a_i, b_i\ge 0$ with $a_i\neq b_i$ such that $c_i(e)\in \{a_i, b_i\}$ for any $e\in M$.
We prove the following theorem.
\begin{theorem}\label{thm:bivalued}
There exists an EFX allocation of chores when $n=3$ and each agent has a personalized bi-valued cost function.
Moreover, we can find an EFX allocation in polynomial time.
\end{theorem}
By scaling each cost function, we assume that $c_i(e)\in \{\varepsilon_i, 1\}$ for any $i\in N (=[3]) \ {\rm and}\ e\in M$, where $\varepsilon_i\in [0, 1)$.
We first give some definitions.
We call chore $e\in M$ a {\it consistently large chore} if $c_i(e)=1$ for all $i\in N$, and a {\it consistently small chore} if $c_i(e)=\varepsilon_i$ for all $i\in N$.
%We call chore $e\in M$ {\it large to agent} $i$ if $c_i(e)=1$, and {\it small to agent }$i$ if $c_i(e)=\varepsilon_i$.
%If a chore is large (resp.~small) to agent $i$, but small (resp. large) to all other agents, we say that it is large (resp.~small) only to agent $i$.
We call chore $e\in M$ a~{\it large} (resp.~{\it small}) {\it chore only for one agent} if $c_i(e)=1$ (resp.~$c_i(e)=\varepsilon_i$) for some agent $i$ and $c_j(e)=\varepsilon_j$ (resp.~$c_j(e)=1$) for any other agents $j\in N\setminus i$.
Note that every chore can be categorized into one of the following four types: consistently large, consistently small, large only for one agent, and  small only for one agent.
%We call an allocation $A=(A_1,\ldots, A_n)$ {\it balanced} if $||A_i|-|A_j||\le 1$ for any $i, j\in [n]$.
\paragraph{Round-Robin Algorithm}
In our algorithm, we use the round-robin algorithm as in~\cite{zhou2021approximately} in several parts.
See Algorithm~\ref{alg:04} for the description of the round-robin algorithm.
Note that the algorithm is now described for general $n$, while it will be used for $n=3$ in this section.  
The algorithm takes as input an ordering of agents $(\sigma_1, \sigma_2, \ldots, \sigma_n)$, a set of chores, and a cost function for each agent.
In the order of $\sigma$, the agents choose the minimum cost chore for her perspective until all chores are allocated. 
We index the rounds by $1,\ldots, m$, where exactly one chore is allocated in each round.
For every agent $i$, denote by $r_i$ the last round in which agent $i$ received a chore.
To simplify the notation, let $r_i=0$ if $i$ received no chore. 
Note that each $r_i$ depends on the ordering $\sigma$ of the agents. 
We call the output of Algorithm~\ref{alg:04} the {\it round-robin allocation with respect to $\sigma$}.
If $\sigma$ is not specified, it is simply called a {\it round-robin allocation}. 

\begin{algorithm}[htb]
\caption{ Round-Robin Algorithm}
\label{alg:04}
\begin{algorithmic}[1]
\Require an ordering of the agents $(\sigma_1,\ldots, \sigma_n)$, a set of chores $M$, and a cost function $c_i$ for each $i \in N$.
\State Initialize $A_i\leftarrow \emptyset$ for all $i\in N$, $R\leftarrow M$, and $i\leftarrow 1$
\While{$R\neq \emptyset$}
    \State Pick up $e \in \argmin_{e'\in R}\{c_{\sigma_i}(e')\}$
    \State $A_{\sigma_i}\leftarrow A_{\sigma_i}\cup e$, $R\leftarrow R\setminus e$
    \State Set $i\leftarrow (i~~{\rm mod}\ n)+1$
\EndWhile
\State \Return $A=(A_1,\ldots, A_n)$
\end{algorithmic}
\end{algorithm}

Note that for a round-robin allocation $A$, it holds that $||A_i|-|A_j||\le 1$ for any $i,j\in N$.
The following lemma is a fundamental property of the output of Algorithm~\ref{alg:04}.
\begin{lemma}[Lemma 5.3.~in~\cite{zhou2021approximately}]\label{lem:ef1}
Let $A$ be an allocation obtained by Algorithm~\ref{alg:04}. For any distinct agents $i, j$ with $r_i<r_j$, we have  
$c_i(A_i)\le c_i(A_j)$ and $c_j(A_j\setminus e)\le c_j(A_i)$ for some $e\in A_j$.
\end{lemma}
%agent $i$ does not envy $j$, and $j$ is EF1 towards agent $i$

\paragraph{Overview}
Our algorithm to find an EFX allocation is described in Algorithm~\ref{alg:03}.
Our basic idea is to recursively compute an EFX allocation with respect to the number of chores. 
%%If $M$ is the empty set, the algorithm returns the empty allocation. 
If $|M| \le 3$, then the algorithm returns an allocation in which every agent receives at most one chore.
Otherwise, we deal with the following two cases separately. 

Suppose first that there exists a chore $e$ which is consistently small or large only for one agent. 
In this case, 
%%for the instance with $e$ removed, 
we compute an EFX allocation $A$ to $N$ of $M \setminus e$ recursively. 
By regarding $A$ as an allocation to $U$, we construct the EFX-graph $G_A$. 
%%Note that $G_A$ has a perfect matching. 
We show that $e$ can be added to a certain vertex $u^*\in U$ so that the resulting EFX-graph has a perfect matching.
The algorithm computes a perfect matching in the new EFX-graph, which gives an EFX allocation to $N$ of $M$.  

Suppose next that there exists neither consistently small chore nor large chore only for one agent, that is, 
the instance consists of consistently large chores and small chores only for one agent.
If all chores are consistently large, then a round-robin allocation gives an EFX allocation in this case. 
If there is a small chore $e$ only for one agent (say agent $1$) and there are no small chores only for agent 2 or 3, then  
the algorithm computes a round-robin allocation $(S_1, S_2, S_3)$ on the instance with $e$ removed such that $r_1< \min\{r_2, r_3\}$.
In this case, we can show that $(S_1\cup e, S_2, S_3)$ is an EFX-allocation of $M$. 
If there are a small chore $e$ only for one agent (say agent $1$) and a small chore $e'$ for another agent (say agent $2$), then 
the algorithm computes the round-robin allocation $(S_1, S_2, S_3)$ on the instance with $e$ and $e'$ removed such that $r_1< r_2 < r_3$.
In this case, we can show that $(S_1\cup e, S_2 \cup e', S_3)$ is an EFX-allocation of $M$. 

%%If there are no consistently small chores and large chores only for one agent, then the instance consists of consistently large chores and small chores only for one agent. If there are no small chores only for one agent, then all chores are consistently large, and a round-robin allocation gives an EFX allocation in this case. If there is at least one small chore $e$ only for one agent (say agent $1$), then we consider two cases depending on the situation in the instance with agents 2,3, and the set of chores $M\setminus e$.
%%If there are no small chores only for agent 2 or 3 in $M\setminus e$, then the algorithm computes a round-robin allocation $A$ of $M\setminus e$, and agent $1$ chooses the most favorite bundle among $A$, and allocates the remaining bundles to agents $2$ and $3$ in an arbitrary way.
%%Otherwise, the algorithm computes the round-robin allocation $(S_1, S_2, S_3)$ on the instance with $e$ and $e'$ removed such that $r_1< r_2 < r_3$.
%%Agent $1$ receives $S_1\cup e$, agent $2$ receives $S_2\cup e'$, agent $3$ receives $S_3$.
%%It can be shown that the resulting allocation becomes an EFX allocation.
\begin{algorithm}[htb]
\caption{Case when $n=3$ and each agent has a personalized bi-valued cost function.}
\label{alg:03}
\begin{algorithmic}[1]
%\Require a set of three agents $N=[3]$, a set of chores $M$, and a personalized bi-valued cost function $c_i$ for each $i \in N$.
%\Ensure  an EFX allocation $A^*$
\Procedure{EFX}{$N$, $M$, $\{c_i\}_{i \in N}$}
\If{$|M| \le 3$}
%%    \State \Return the empty allocation $(\emptyset, \emptyset, \emptyset)$
    \State \Return an allocation in which every agent receives at most one chore. 
\Else
    %\State $e_i\leftarrow \argmin_{e\in M} c_i(e)$ for any $i\in N$
    \If{there exists a chore $e$ which is consistently small or large only for one agent}
        \State $A\leftarrow$ {\sc EFX}($N$, $M\setminus e$, $\{c_i\}_{i\in N}$)
        %\State Identify $A$ with $(A_u)_{u\in U}$ be an allocation to $U$, where $U=N$.
        \State Regard $A$ as an allocation to $U$ with $|U|=3$.
        \State Find a vertex $u^*\in U$ such that $G_{A'}$ has a perfect matching,
        \State where $A'$ is the allocation to $U$ obtained from $A$ by adding $e$ to $A_{u^*}$.
        \State $A_{u^*} \leftarrow A_{u^*} \cup e$
        \State Find a perfect matching $X$ on $G_{A}$.
        \State Construct the allocation $A^*$ by allocating each chore set to the matched agent in $X$.
        \State \Return $A^*$
    
    \ElsIf{there are no small chores only for one agent in $M$}
            \State \Return a round-robin allocation $A^*$ of $M$ with respect to an arbitrary ordering
        \Else%{ there exists a...chore in $M$}
            \State We assume that there exists $e\in M$ such that $c_1(e)=\varepsilon_1$ and $c_2(e)=c_3(e)=1$, 
            \State renumbering if necessary.
            \If{there are no small chores only for agent $2$ or $3$ in $M\setminus e$}
                \State Compute a round-robin allocation $S=(S_1, S_2, S_3)$ on $M\setminus e$ s.t. $r_1< \min \{r_2, r_3 \}$.
                \State Set $A^*_1\leftarrow S_1\cup e$, $A^*_2\leftarrow S_2$, $A^*_3\leftarrow S_3$.
                \State \Return $A^*$
            \Else
                \State We assume that there exists $e'\in M\setminus e$ s.t. $c_1(e')=c_3(e')=1$ and $c_2(e')=\varepsilon_2$.
                \State Compute a round-robin allocation $S=(S_1, S_2, S_3)$ on $M\setminus \{e, e'\}$ s.t. $r_1 < r_2 < r_3$.
                \State Set $A^*_1\leftarrow S_1\cup e$, $A^*_2\leftarrow S_2\cup e'$, $A^*_3\leftarrow S_3$.
                \State \Return $A^*$
            \EndIf 
    \EndIf
\EndIf
\EndProcedure
\end{algorithmic}
\end{algorithm}
\begin{proof}[Proof of Theorem~\ref{thm:bivalued}]
We first show the correctness of Algorithm~\ref{alg:03} by induction on $m$.
For the base case of $m\le 3$, the algorithm trivially returns an EFX allocation.
For the inductive step, we assume that the algorithm returns an EFX allocation when the number of chores is less than $m$.
\begin{description}
    \item[Case 1:] There exists a chore $e$ which is consistently small or large only for one agent in $M$.\\
    Let $A=(A_i)_{i\in N}$ be an EFX allocation of $M\setminus e$, whose existence is guaranteed by the induction hypothesis. 
    We regard $A$ as an allocation to $U$ by using an arbitrary bijection between $N$ and $U$.
    Note that EFX-graph $G_A$ has a perfect matching since $A=(A_i)_{i\in N}$ is an EFX allocation.
%%    We consider the following two cases separately. 
    \begin{description}
        \item[Case 1-1:] $e$ is a consistently small chore.\\
        By Observation~\ref{ob:simple} (iii), there exists a perfect matching $X$ in $G_A$ such that $(i, X(i))\in E^{\min}_{A}$ for some $i\in N$.
        Since $e$ is a consistently small chore, we have $c_i(e)=\varepsilon_i \le c_i(e')$ for any $e'\in M\setminus e$.
        Thus, by applying Corollary~\ref{cor:1}, 
        $G_{A'}$ has a perfect matching $X'$, where $A'$ is the allocation to $U$ of $M$ obtained from $A$ by adding $e$ to $A_{X(i)}$.
        %Therefore, the algorithm returns an EFX allocation of $M$ by Observation~\ref{ob:final}.
        \item[Case 1-2:] $e$ is a large chore only for one agent.\\
        We assume that $c_1(e)=\varepsilon_1, c_2(e)=\varepsilon_2$, and $c_3(e)=1$, renumbering if necessary.
        This means that $c_i(e)\le c_i(e')$ for any $i\in \{1,2\}$ and $e'\in M\setminus e$.
        Thus, by applying Lemma~\ref{cl:add} for $M'=M\setminus e$, there exists a vertex $u^*\in U$ such that $G_{A'}$ has a perfect matching, where $A'$ is the allocation to $U$ of $M$ obtained from $A$ by adding $e$ to $A_{u^*}$.
        %Therefore, the algorithm returns an EFX allocation of $M$ by Observation~\ref{ob:final}.
    \end{description}
    In both cases, $G_{A'}$ has a perfect matching.
    Therefore, the algorithm returns an EFX allocation by Observation~\ref{ob:final}. 

\begin{figure}[t]
 \begin{tabular}{ccc}
 \begin{minipage}[b]{0.3\hsize}
  \begin{center}
     \begin{tabular}{c|cccc}
                        & $e_1$ & $e_2$ & $e_3$ &$\cdots$ \\ \hline
                agent 1 & 1     & 1     & 1 & $\cdots$\\
                agent 2 & 1     & 1     & 1 & $\cdots$\\
                agent 3 & 1     & 1     & 1 & $\cdots$\\
        \end{tabular}
    \tblcaption{Situation in Case 2}
    \label{tab:1}
   \end{center}
 \end{minipage}
 \hfill
\begin{minipage}[b]{0.3\hsize}
  \begin{center}
    \begin{tabular}{c|cccc}
                        & $e$ & $\cdots$ &  & \\ \hline
                agent 1 & $\varepsilon_1$     &  &  &$\cdots$ \\
                agent 2 & 1     & 1     & 1 & $\cdots$\\
                agent 3 & 1     & 1     & 1 & $\cdots$\\
            \end{tabular}
    \tblcaption{Situation in Case 3-1}
    \label{tab:2}
   \end{center}
 \end{minipage}
 \hfill
\begin{minipage}[b]{0.3\hsize}
  \begin{center}
    \begin{tabular}{c|ccc}
                        & $e$ & $e'$ & $\cdots$ \\ \hline
                agent 1 & $\varepsilon_1$     &  1    &$\cdots$\\
                agent 2 & 1     & $\varepsilon_2$     &$\cdots$\\
                agent 3 & 1     & 1     & $\cdots$ \\
            \end{tabular}
    \tblcaption{Situation in Case 3-2}
    \label{tab:3}
   \end{center}
 \end{minipage}
 \end{tabular}
\end{figure}
    \item[Case 2:] All chores in $M$ are consistently large (see Table~\ref{tab:1}).\\
    In this case, any round-robin allocation $A^*$ gives an EFX allocation since it satisfies that $||A^*_i|-|A^*_j||\le 1$ for any $i,j\in N$.

    \item[Case 3:] Otherwise.\\
    In this case, $M$ consists of consistently large chores and small chores only for one agent, and there is at least one small chore only for one agent.
    We assume that there exists $e\in M$ such that $c_1(e)=\varepsilon_1$ and $c_2(e)=c_3(e)=1$, renumbering if necessary.
    \begin{description}
        \item[Case 3-1:] There are no small chores only for agent $2$ or $3$ in $M\setminus e$ (see Table~\ref{tab:2}).\\
        In this case, all chores in $M\setminus e$ are consistently large or small only for agent $1$.
        This means that $c_2(e')=c_3(e')=1$ for any $e'\in M\setminus e$.
        Let $S=(S_1, S_2, S_3)$ be a round-robin allocation of $M\setminus e$ such that $r_1< \min\{ r_2, r_3\}$.
        Note that we can achieve $r_1< \min\{ r_2, r_3\}$ by choosing the ordering $\sigma$ of the agents appropriately, 
        because $M \setminus e$ contains at least three chores. 
        We show that $A^*=(A^*_1, A^*_2, A^*_3)$ is an EFX allocation of $M$, where $A^*_1= S_1\cup e$, $A^*_2= S_2$, $A^*_3= S_3$.
        By Lemma~\ref{lem:ef1}, we have $c_1(S_1)\le \min \{ c_1(S_2), c_1(S_3) \}$.
        Thus, it holds that $\max_{f\in A^*_1} c_1(A^*_1 \setminus f)=c_1(S_1) \le \min \{ c_1(S_2), c_1(S_3) \} = \min \{ c_1(A^*_2), c_1(A^*_3) \}$, which implies that agent $1$ does not strongly envy anyone.
        We also see that $\max_{f\in A^*_2}c_2(A^*_2\setminus f)=  |S_2|-1 \le |S_3|=c_2(A^*_3)$ and $\max_{f\in A^*_3}c_3(A^*_3\setminus f)= |S_3|-1 \le |S_2|=c_3(A^*_2)$, because $||S_2|-|S_3||\le 1$. 
        Thus, agents $2$ and $3$ do not strongly envy each other.
        Finally, for $i\in \{2,3\}$,  we have $ c_i(A^*_i)=|S_i| \le |S_1|+1=c_i(A^*_1\cup e)$, which implies that agents $2$ and $3$ do not envy agent $1$.
        Therefore, $A^*=(A^*_1, A^*_2, A^*_3)$ is an EFX allocation of $M$.
        \item[Case 3-2:] There is a small chores only for agent $2$ or $3$ in $M\setminus e$.\\
        In this case, we assume that there exists $e'\in M\setminus e$ such that $c_1(e')=c_3(e')=1$ and $c_2(e')=\varepsilon_2$, renumbering if necessary (see Table~\ref{tab:3}).
        Let $S=(S_1, S_2, S_3)$ be a round-robin allocation on $M\setminus \{e, e'\}$ such that $r_1 < r_2 < r_3$.
        Note that we can achieve $r_1 < r_2 < r_3$ by choosing the ordering $\sigma$ of the agents appropriately, 
        because $M \setminus \{e, e'\}$ contains at least two chores. 
        We show that $A^*=(A^*_1, A^*_2, A^*_3)$ is an EFX allocation of $M$, where $A^*_1= S_1\cup e$, $A^*_2= S_2\cup e'$, $A^*_3= S_3$.
        We have $\max_{f\in A^*_1} c_1(A^*_1 \setminus f)=c_1(S_1)\le \min \{c_1(S_2), c_1(S_3) \} \le \min \{c_1(A^*_2), c_1(A^*_3) \}$ by Lemma~\ref{lem:ef1}, which implies that agent 1 does not strongly envy anyone.
        We also have $\max_{f\in A^*_2} c_2(A^*_2 \setminus f)=c_2(S_2)\le \min \{c_2(S_1)+1, c_2(S_3)\}=\min \{c_2(A^*_1), c_2(A^*_3)\}$, where the inequality follows from Lemma~\ref{lem:ef1}.
        This means that agent 2 does not strongly envy anyone.
        Finally, we have $c_3(A^*_3) = c_3(S_3) \le \min \{c_3(S_1)+1, c_3(S_2)+1\}=\min\{c_3(A^*_1), c_3(A^*_2)\}$, where the inequality follows from Lemma~\ref{lem:ef1}.
        This means that agent 3 does not envy anyone.
        Therefore, $A^*=(A^*_1, A^*_2, A^*_3)$ is an EFX allocation of $M$.
    \end{description}
\end{description}
We next show that Algorithm~\ref{alg:03} runs in polynomial time.
Each if-statement can be checked in $O(m)$ time.
Finding a perfect matching in $G_{A'}$ or constructing an allocation from a perfect matching can be done in polynomial time.
The round-robin algorithm also runs in polynomial time.
Therefore, Algorithm~\ref{alg:03} runs in polynomial time. 
\end{proof}

\section*{Acknowledgments}
This work was partially supported by the joint project of Kyoto University and Toyota Motor Corporation, titled ``Advanced Mathematical Science for Mobility Society'' and by JSPS KAKENHI Grant Number
JP20K11692. %% Yusuke Kobayashi KIBAN C.

\bibliographystyle{plain}
\bibliography{main}

\end{document}